\documentclass{stacs_proc}
\usepackage[latin1]{inputenc}
\usepackage{graphicx,xspace}
\usepackage{epsfig}
\usepackage{color}
\usepackage{amsmath}
\usepackage{amssymb}
\usepackage{comment}
\usepackage{subfigure}
\usepackage{tikz}
\usepackage[english]{babel}

\newcommand\Z{\mathbb{Z}}
\newcommand\N{\mathbb{N}}
\newcommand\Tilings{{\mathcal T}_\tau}
\newcommand\ForbPat{{\mathcal F}_\tau}
\newcommand\AllowPat{{\mathcal P}_\tau}

\newcommand\rk{\rho}
\newcommand\ie{i.e.,\ }
\newcommand\eg{e.g.\ }
\newcommand\typea{{\textit{type~a}}\xspace}
\newcommand\typeb{{\textit{type~b}}\xspace}
\newcommand\OP{{\mathcal O}_P}
\newcommand\class[1]{\ensuremath{\left<#1\right>}}

\begin{document}

\title[Structural aspects of tilings]{Structural aspects of tilings}
\author[]{A. Ballier}{Alexis Ballier}
\author[]{B. Durand}{Bruno Durand}
\author[]{E. Jeandel}{Emmanuel Jeandel}

\address[]{Laboratoire d'informatique fondamentale de Marseille (LIF)
Aix-Marseille Universit\'e, CNRS\newline                                                
39 rue Joliot-Curie, 13\,453 Marseille Cedex 13, France}                    
\email{alexis.ballier@lif.univ-mrs.fr}
\email{bruno.durand@lif.univ-mrs.fr}
\email{emmanuel.jeandel@lif.univ-mrs.fr}

\keywords{tiling, domino, patterns, tiling preorder, tiling structure}
\subjclass{G.2.m}

\begin{abstract}
    In this paper, we study the structure of
    the set of tilings produced by any given tile-set. 
    For better understanding this structure, we address the set of finite
    patterns that each tiling contains.

    This set of patterns can be analyzed in two different contexts: the first one is
    combinatorial and the other topological. These two approaches have
    independent merits and, once combined, provide somehow surprising
    results.

    The particular case where the set of produced tilings is countable is
    deeply investigated while we prove that the uncountable case may have a
    completely different structure.

    We introduce a pattern preorder and also make use of Cantor-Bendixson rank.
    Our first main result is that a tile-set that produces only periodic tilings
    produces only a finite number of them. Our second main result exhibits a tiling with
    exactly one vector of periodicity in the countable case.
\end{abstract}

\maketitle

\stacsheading{2008}{61-72}{Bordeaux}
\firstpageno{61}

\section{Introduction}

Tilings are basic models for geometric phenomena of computation: local
constraints they formalize have been of broad interest in the
community since they capture geometric aspects of computation
\cite{robinson,Ber-undecidability-dp,DBLP:journals/jsyml/Hanf74,
DBLP:journals/jsyml/Myers74,DBLP:conf/stoc/DurandLS01}.  This
phenomenon was discovered in the sixties when tiling problems happened
to be crucial in logic: more specifically, interest shown in tilings
drastically increased when Berger proved the undecidability of the
so-called domino problem~\cite{Ber-undecidability-dp} (see also
\cite{gurkor72} and the well known book
\cite{DBLP:books/sp/BorgerGG1997} for logical aspects).  Later,
tilings were basic tools for complexity theory (see the nice review of
Peter van Emde Boas \cite{forever} and some of Leonid Levin's paper
such as \cite{DBLP:journals/siamcomp/Levin86}).

Because of growing interest for this very simple model, several
research tracks were aimed directly on tilings: some people tried to
generate the most complex tilings with the most simple constraints
(see
\cite{robinson,DBLP:journals/jsyml/Hanf74,DBLP:journals/jsyml/Myers74,DBLP:conf/stoc/DurandLS01}),
 while others were most interested in structural aspects (see
\cite{radin92space,DBLP:journals/tcs/Durand99}).

In this paper we are interested in structural properties of tilings.
We choose to focus on finite patterns tilings contain and thus introduce a
natural preorder on tilings: a tiling is extracted
from another one if all finite patterns that appear in the first one
also appear in the latter. We develop this combinatorial notion in
Section \ref{sec:def:tilings}. This approach can be expressed in terms
of topology (subshifts of finite type) and we shall explain the relations
between both these approaches in Section \ref{sec:def:topo}.

It is important to stress that both these combinatorial and
topological approaches have independent merits.  Among the results we
present, different approaches are indeed used for proofs.  More specifically,
our first main result (Theorem~\ref{thm:peralorsperfini}) states that if a
tile-set produces only periodic tilings then it produces only finitely many of
them;
despite its apparent simplicity, we did not find any proof of Theorem~\ref{thm:peralorsperfini}
in the literature. 
Our other main result (Theorem~\ref{thm:dnbalors1perstrict})
which states that in the countable case a tiling with exactly one vector of
periodicity exists is proved
with a strong help of topology.

Our paper is organized as follows: Section~\ref{sec:def} is devoted to
definitions (combinatorics, topology) and basic structural remarks.
In Section~\ref{sec:main} we prove the existence of minimal and maximal 
elements in tilings enforced by a tile-set.  Then we present an analysis in terms of
Cantor-Bendixson derivative which provides powerful tools.  We study
the particular case where tilings are countable and present our main
results. We conclude by some open problems.

\section{Definitions}
\label{sec:def}
\subsection{Tilings}\label{sec:def:tilings}
We present notations and definitions for tilings since several models
are used in literature: Wang tiles, geometric frames of rational
coordinates, local constraints\ldots All these models are equivalent
for our purposes since we consider very generic properties of them
(see \cite{DBLP:journals/tcs/CervelleD04} for more details and proofs).
We focus our study on tilings of the plane although our results hold in higher
dimensions.

In our definition of tilings, we first associate a state to each
cell of the plane. Then we impose a local constraint on them.  More
formally, $Q$ is a finite set, called the \emph{set of states}.  A
\emph{configuration} $c$ consists of cells of the plane
with states, thus $c$ is an element of $Q^{\Z^2}$.  We denote by
$c_{i,j}$ or $c(i,j)$  the state of $c$ at the cell $(i,j)$.

A tiling is a configuration which satisfies a given finite set of
finite constraints everywhere.  More specifically we express these constraints
as a set of allowed patterns: a configuration is a tiling if around any of its
cells we can see one of the allowed patterns:

\begin{definition}[patterns]
A \emph{pattern} $P$ is a finite restriction of a configuration \ie 
an element of $Q^V$ for some finite domain $V$ of $\Z^2$. A pattern 
\emph{appears} in a configuration $c$ (resp.\ in some other pattern 
$P'$) if it can be found somewhere in $c$ (resp.\ in $P'$); \ie
if there exists a vector $t \in \Z^2$ such that $c(x+t) = P(x)$ on 
the domain of $P$
(resp.\ if $P'(x+t)$ is defined for $x\in V$ and $P'(x+t) = P(x)$) .
\end{definition}

By language extension we say that  a
pattern is \emph{absent} or \emph{omitted} in a configuration if it does not 
appear in it.

\begin{definition}[tile-sets and tilings]
A \emph{tile-set} is a tuple $\tau=(Q, \AllowPat)$ where $\AllowPat$ is a finite
set of patterns on $Q$.
All the elements of $\AllowPat$ are supposed to be defined on the same domain 
denoted by $V$ ($\AllowPat \subseteq Q^V$).

A \emph{tiling} by $\tau$ is a configuration $c$ equal to one of 
the patterns on all cells:
$$\forall x \in \Z^2, c|_{V+x} \in \AllowPat$$

We  denote by $\Tilings$ the set of tilings by $\tau$. 
\end{definition}

Notice that in the definition of one tile-set we can allow patterns of
different definition domains provided that there are a finite number of them.



An example of a tile-set defined by its allowed patterns is given in Fig.~\ref{fig:tuiles}.
The produced tilings are given in Fig.~\ref{fig:treillis}; the meaning of
  the edges in the graph will be explained later; tilings are represented modulo shift.
In $A_i$ and $B_i$, $i$ is an integer that represents the size
of the white stripe.

\begin{figure}[ht]
\centering
\begin{tikzpicture}[scale=.5]
\filldraw[fill=red] (0,7) rectangle +(1,1);
\filldraw[fill=red] (1,7) rectangle +(1,1);
\filldraw[fill=red] (3,7) rectangle +(1,1);
\filldraw[fill=white] (4,7) rectangle +(1,1);
\filldraw[fill=red] (6,7) rectangle +(1,1);
\filldraw[fill=green] (7,7) rectangle +(1,1);
\filldraw[fill=red] (9,7) rectangle +(1,1);
\filldraw[fill=black] (10,7) rectangle +(1,1);
\filldraw[fill=white] (1,5) rectangle +(1,1);
\filldraw[fill=white] (2,5) rectangle +(1,1);
\filldraw[fill=green] (4,5) rectangle +(1,1);
\filldraw[fill=green] (5,5) rectangle +(1,1);
\filldraw[fill=black] (7,5) rectangle +(1,1);
\filldraw[fill=black] (8,5) rectangle +(1,1);
\filldraw[fill=red] (0,3) rectangle +(1,1);
\filldraw[fill=green] (2,3) rectangle +(1,1);
\filldraw[fill=green] (4,3) rectangle +(1,1);
\filldraw[fill=white] (6,3) rectangle +(1,1);
\filldraw[fill=white] (8,3) rectangle +(1,1);
\filldraw[fill=black] (10,3) rectangle +(1,1);
\filldraw[fill=red] (0,2) rectangle +(1,1);
\filldraw[fill=green] (2,2) rectangle +(1,1);
\filldraw[fill=white] (4,2) rectangle +(1,1);
\filldraw[fill=white] (6,2) rectangle +(1,1);
\filldraw[fill=black] (8,2) rectangle +(1,1);
\filldraw[fill=black] (10,2) rectangle +(1,1);
\end{tikzpicture}
\caption{Allowed patterns}
\label{fig:tuiles}
\end{figure}
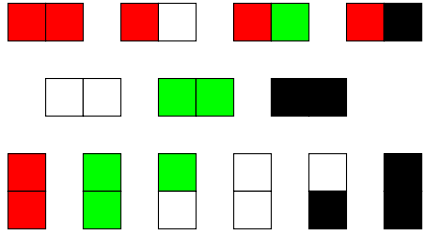

\begin{figure}[ht]
\centering
\includegraphics{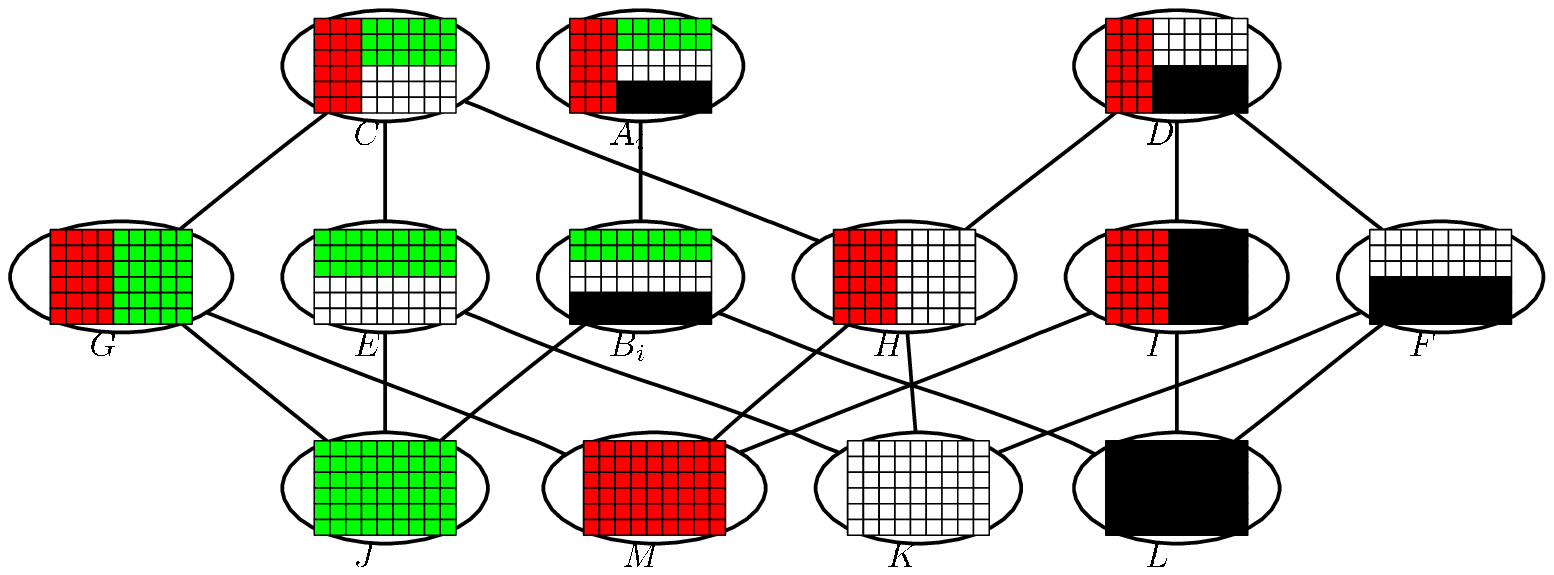}

{\scriptsize An edge represents a relation
$Q \prec P$ if $P$ is above $Q$. Transitivity edges are not depicted. 
As an example $K \prec E$ and $K \prec C$.}
\caption{Hasse diagram of the order $\prec$ with the tile-set defined in
Fig.~\ref{fig:tuiles}}
\label{fig:treillis}
\end{figure}

Throughout the following, it will be more convenient for us to define tile-sets by
the set of their \emph{forbidden patterns}:
a tile-set is then given 
by a finite set $\ForbPat$ of forbidden patterns $(\ForbPat = Q^V 
\setminus
\AllowPat)$; a configuration is a
tiling if no forbidden pattern appears.

Let us now introduce the following natural preorder, which will play a central
role in our paper:

\begin{definition}[Preorder]
\label{def:preorder}
Let $x,y$ be two tilings, we say that $x\preceq y$ if any pattern that appears
in $x$ also appears in $y$.
\end{definition}

We say that two tilings $x,y$ are equivalent if $x\preceq y$ and
$y\preceq x$.  We denote this relation by $x\approx y$.  In this case,
$x$ and $y$ contain the same patterns.  The equivalence class of $x$
is denoted by $\class{x}$.  We  write $x \prec y$ if $x \preceq
y$ and $x \not\approx y$.

Some structural properties of tilings can be seen with the help of
this preorder.  The Hasse diagram in Fig.~\ref{fig:treillis} correspond to the
relation $\prec$.



We choose to distinguish two types of tilings: A tiling $x$ is of
\typea if any pattern that appears in $x$ appears infinitely many
times; $x$ is of \typeb if there exists a pattern that appears only
once in $x$.  Note that any tiling is either of \typea or of \typeb:
suppose that there is a pattern that appears only a finite number of
times in $x$; then the pattern which is the union of those
patterns appears only once.

If $x$ is of \typeb, then the only tilings equivalent to $x$ are its shifted:
there is a unique way in $\class{x}$ to grow around the unique pattern.

\subsection{Topology}
\label{sec:def:topo}



In the domain of symbolic dynamics, topology provides both interesting
results and is also a nice condensed way to express some combinatorial proofs
\cite{hedlund,gohe}.  The benefit of topology is a
little more surprising for tilings since they are essentially static
objects.  Nevertheless, we can get nice results with topology as will
be seen in the sequel.

We see the space of configurations $Q^{\Z^2}$ as a metric space in the following way:
the distance between two configurations $c$ and $c'$ is $2^{-i}$ where
$i$ is the minimal offset (for \eg the euclidean norm) of a point where $c$ and $c'$ differ:
\[
	d(c,c') = 2^{-\min  \{ |i|,\ c(i) \not= c'(i) \}}
 \]
We could also endow $Q$ with the discrete topology and then $Q^{\Z^2}$
with the product topology, thus obtaining the same topology as the one induced by $d$.

In this topology, a basis of open sets is given through the patterns: 
for each pattern $P$, the set $\OP$ of all configurations $c$ which contains
$P$ in their center (\ie such that $c$ is equal to $P$ on its domain) is an open
set, usually called a cylinder. 
Furthermore cylinders such defined are also closed 
(their complements are finite unions of  $\mathcal{O}_{P'}$ where $P'$ are patterns of same 
domain different from $P$). Thus $\OP$'s are clopen.

\begin{proposition}
$Q^{\Z^2}$ is a compact perfect metric space (a Cantor space).
\end{proposition}
We say that a set of configurations $S$ is shift-invariant if any
shifted version of any of its configurations is also in $S$; \ie if for
every $c \in S$, and every $v \in \Z^2$ the configuration $c'$ defined
by $c'(x) = c(x+v)$ is also in $S$. We denote such a shift by $\sigma_v$.

\begin{remark}
\label{remark:preordertopo}
Our definition of pattern preorder \ref{def:preorder} can be reformulated in a topological way :
$x \preceq y$ if and only if there exists shifts
$(\sigma_i)_{i\in\N}$ such that $\displaystyle \sigma_i(y)\xrightarrow[i\to\infty]{} x$.
We say that $x$ can be extracted from $y$.
\end{remark}

For a given configuration $x$, we define the topological closure of shifted
forms of $x$:
$\Gamma (x) = \overline{\{ \sigma_v(x),\ v \in \Z^2
\}}$ where $\sigma_{i,j}$ represents a shift of vector $v$.

We see that $x \preceq y$ if and only if $\Gamma (x) \subseteq \Gamma
(y)$. Remark that 
$x$ is minimal for $\prec$ if and only if $\class{x}$ is closed.

As sets of tilings can be defined by a finite number of forbidden patterns, 
they correspond to \emph{subshifts} of finite type\footnote{
\emph{Subshifts} are closed shift-invariant subsets of $Q^{\Z^2}$}.
In the sequel, we sometimes use arbitrary
subshifts; they correspond to a set of configurations
with a potentially infinite set of forbidden patterns.

\section{Main results}
\label{sec:main}
\subsection{Basic structure}


Let us first present a few structural results.
First, the existence of minimal classes for $\prec$ is well known.  


\begin{theorem}[minimal elements]
Every set of tilings contains a minimal class for $\prec$. 
\end{theorem}

In the context of tilings, those that belong to minimal classes are 
often called \emph{quasiperiodic}, while in language theory they are 
called \emph{uniformly recurrent} 
or \emph{almost periodic}. 
Those quasiperiodic configurations admit a nice characterization: any
pattern that appears in one of them can be found in any sufficiently
large pattern (placed anywhere in the configuration).

For a combinatorial proof of this theorem see
\cite{DBLP:journals/tcs/Durand99}.  Alternatively, 
here is a scheme of a topological proof: 
consider a minimal
subshift of $\Tilings$ (such a subshift exists, see \eg
\cite{radin92space}) then every tiling in this set is in a
minimal class.


An intensively studied class of tilings is the set of self-similar tilings.
These tilings indeed are minimal elements (quasiperiodic) but one can find other
kinds of minimal tilings (\eg the nice approach of Kari and Culik in \cite{
culik97aperiodic}).

The existence of maximal classes of tilings is not trivial and we
have to prove it:
\begin{theorem}[maximal elements]
Every set of tilings  contains a maximal class for $\prec$.
\end{theorem}

\begin{proof}
	Let us prove that any increasing chain has a least upper bound. The theorem
	is then obtained by Zorn's lemma.
	
    Consider $T_i$ an increasing chain of tiling classes.  Consider the
    set $P$ of all patterns that this chain contains.  As the set of
    all patterns is countable, $P$ is countable too, $P=\{p_{i}\}_{i\in \N}$.

    Now consider two tilings $T_k$ and $T_l$, any pattern that appears in $T_k$
    or $T_l$ appears in $T_{\max(k,l)}$.
    Thus we can construct a sequence of patterns $(p'_i)_{i\in\N}$ such that
    $p'_i$
    contains all $p_{j}$, $j\leq i$ and $p'_{i-1}$.  
    Note that $p'_{i}$ is correctly tiled by the considered tile-set.
    
    The sequence of patterns $p'_{i}$ grows in size.  By shift
    invariance, we can center each $p'_{i}$ by
    superimposing an instance of $p'_{i-1}$ found in $p'_{i}$ over
    $p'_{i-1}$. We can conclude that this sequence has a limit and
    this limit is a tiling that contains all $p_i$, hence is an upper
    bound for the chain $T_i$.  
\end{proof}

Note that this proof also works when the \emph{set of states} $Q$
and/or the \emph{set of forbidden patterns} $\ForbPat$ are countably
infinite (neither compactness nor finiteness is assumed).
However it is easy to construct examples where $Q$ is infinite and there does not exist a minimal tiling.

Note that we actually prove that every chain has not only a upper bound, but
also a least upper bound.
Such a result does not hold for lower bound: We can easily build chains
with lower bounds but no greatest lower bound.

\subsection{Cantor-Bendixson}
In this section we use the topological derivative and define
Cantor-Bendixson rank; then we discuss properties of sets of tilings from this
viewpoint. Most of the results presented in this section are direct translations
of well known results in topology \cite{kura}.

A configuration $c$ is said to be \emph{isolated} in a set of configurations $S$ if
there exists a pattern $P$ (of domain $V$) such that $c$ is the only configuration in 
$S$ that contains the pattern $P$ in its center 
($\forall x \in V, c(x) = P(x)$). We say that $P$ \emph{isolates} $c$.
This corresponds to the topological notion: 
a point is isolated if there exists an open set that contains only this point.
As an example, in Fig.~\ref{fig:cantorbendixsonrang}, the tilings $A_i$ are
isolated, the pattern isolating an $A_i$ is the boundary between red, white,
black and green parts of it.


The topological derivative of a set $S$ is formed by its elements that are not
isolated. We denote it by $S'$.

If $S$ is a set of tilings, or more generally a subshift, we get some more
properties. 
If $P$ isolates a configuration in $S$ then a shifted form of $P$ isolates a
shifted form of this configuration.
Any configuration of $S$ that
contains $P$ is isolated.

As a consequence, if $S = \Tilings$,
then $S' = {\mathcal T}_{\tau'}$ where 
$\tau'$ forbids the set\\ ${\ForbPat \cup \{P | P \text{\ isolates some configuration in\ }
\Tilings\}}$.

Note that $S'$ is not always a set of tilings, but remains a subshift.
Let us examine the example shown in Fig.~\ref{fig:cantorbendixsonrang}.
$S'$ is $S$ minus the classes $A_{i}$. However any set of tilings (subshift of
finite type) that contains $C,B_i$ and $D$ also contains $A_i$.
Hence $S'$ is not of finite type in this example.

We define inductively $S^{(\lambda)}$ for any ordinal $\lambda$ :
\begin{itemize}
\item $S^{(0)}=S$
\item $S^{(\alpha +1)}=(S^{(\alpha )})'$
\item $S^{(\lambda)}=\bigcap_{\alpha <\lambda}S^{(\alpha)}$ when $\lambda$ is a
limit ordinal.
\end{itemize}

Notice that there exists a countable ordinal $\lambda$ such that $S^{(\lambda +1)} = S^{(\lambda)}$.
Indeed, at each step of the induction, the set of forbidden patterns
increases, and there is at most countably many patterns.
We call the least such ordinal the \emph{Cantor-Bendixson rank of $S$}
\cite{kura}.

An element $c$ is \emph{of rank $\lambda$} in $S$ if $\lambda$ is the least ordinal
such that $c\not\in S^{(\lambda)}$.
If no such $\lambda$ exists, $c$ is of infinite rank. 
For instance all strictly quasiperiodic configurations (quasiperiodic configurations that
are not periodic)  are of infinite rank.
We write $\rk(x)$ the rank of $x$.

An example of what Cantor-Bendixson ranks look like is shown in Fig.~\ref{fig:cantorbendixsonrang},
the first row contains the tilings of rank $1$, the second row the ones of rank
$2$ etc.

\begin{figure}[ht]
\begin{center}

\includegraphics{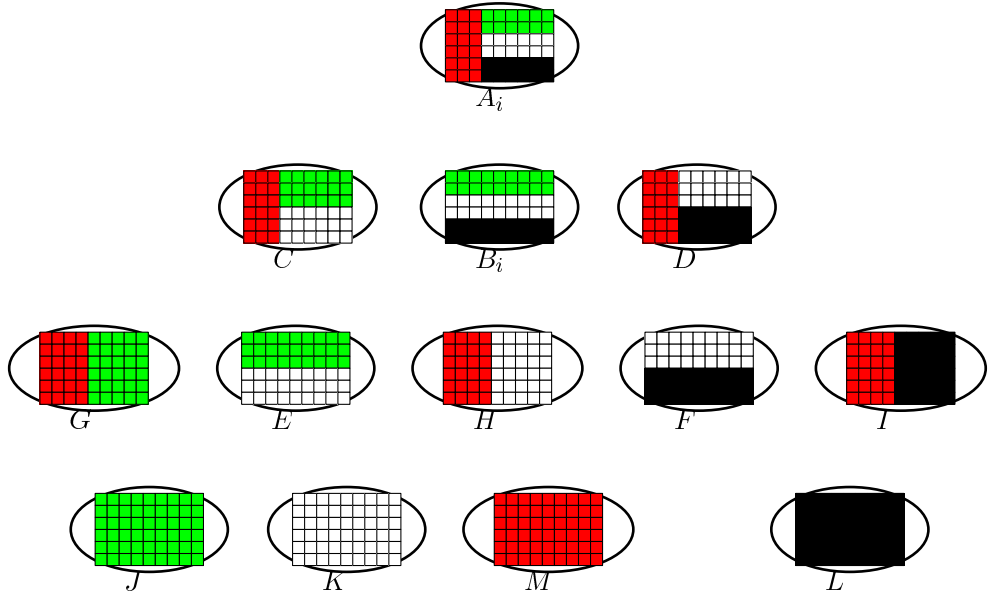}
\end{center}
\caption{Cantor-Bendixson ranks}
\label{fig:cantorbendixsonrang}
\end{figure}

Ranked tilings have many interesting properties. First of all, as any
$\Tilings^{(\lambda)}$ is shift-invariant, a tiling has the same rank as its
shifted forms.

Note that at each step of the inductive definition, the set of isolated
points is at most countable (there are less isolated points than patterns).
As a consequence, if all tilings are ranked, $\Tilings$ is countable, as a
countable union (the Cantor-Bendixson rank is countable) of countable sets.

The converse is also true:
\begin{theorem}
\label{thm:countableequivallranked}
$\Tilings$ is countable if and only if all tilings are ranked.
\end{theorem}
\begin{proof}
Let $\lambda$ be the Cantor-Bendixson rank of $\Tilings$.
$\Tilings^{(\lambda)} = \Tilings^{(\lambda+1)}$ is a perfect set (no points are isolated).
As a consequence, $\Tilings^{(\lambda)}$ must be either empty or uncountable (classical
application of Baire's Theorem : $\Tilings^{(\lambda)}$ is compact thus has the
Baire property and a non empty perfect set with the Baire property cannot be
countable).

As $\Tilings$ is countable, $\Tilings^{(\lambda)}= \emptyset$.
\end{proof}

\begin{remark}
\label{rq:minperio}
Strictly quasiperiodic tilings only appear when the number of possible tilings
is uncountable \cite{DBLP:journals/tcs/Durand99}. As a consequence, if all tilings are
ranked, strictly quasiperiodic tilings do not appear, thus all minimal tilings are
periodic. In this case we therefore may expect all tilings to be somehow simple.
We'll study this case later in this paper.
\end{remark}

As the topology of $Q^{\Z^2}$ has a basis of clopens $\OP$,
$Q^{\Z^2}$ is a $0$-dimensional space,
thus any subset of $Q^{\Z^2}$ is also $0$-dimensional.
As any (non empty) perfect $0$-dimensional compact metric space is isomorphic to
the Cantor Space we obtain:

\begin{theorem}[Cardinality of tiling spaces]
A set of tilings is either finite, countable or has the cardinality of continuum.
\end{theorem}

Note that the proof of this result does not make use of the continuum hypothesis.

We now present the connection between our preorder $\prec$ and the Cantor-Bendixson rank.

\begin{prop}
Let $x$ and $y$ be two ranked tilings
such that $x\prec y$. Then $\rk(x)>\rk(y)$.
\end{prop}
\begin{proof}
By definition of $\prec$, any pattern that appears in $x$ also appears in $y$.
As a consequence, if $P$ isolates $x$ in $S^{(\lambda)}$, then $x$ is the only
tiling of $S^{(\lambda)}$ that contains $P$ hence $y$ cannot be
in $S^{(\lambda)}$.
\end{proof}
Thus tilings of Cantor-Bendixson rank $1$ (minimal rank) are
maximal tilings for $\prec$. 
Conversely if all tilings are ranked, tilings 
of maximal rank exist and are minimal tilings. These tilings are periodic, see
remark \ref{rq:minperio}.

Another consequence is that if all tilings are ranked, 
no infinite increasing chain for $\prec$ exists because
such chain would induce an infinite decreasing chain of ordinals:

\begin{theorem}
If $\Tilings$ is countable, there is no infinite increasing chain for $\prec$.
\end{theorem}

\subsection{The countable case}

In the context of Cantor-Bendixson ranks,
the case of countable tilings was revealed as an important particular case.
Let us study this case in more details.

If the number of tilings is finite, the situation is easy:
any tiling is periodic.
Our aim is to prove that in the countable case, there exists a tiling
$c$ which has exactly one vector of periodicity (such a tiling is sometimes called weakly
periodic in the literature).

We split the proof in three steps : 
\begin{itemize}
	\item There exists a tiling which is not minimal;
	\item There exists a tiling $c$ which is at level
	  $1$, that is such that all tilings less than $c$ are minimal;
	\item Such a tiling has exactly one vector of periodicity.
\end{itemize}	
The first step is a result of independent interest.
To prove the last two steps we use Cantor-Bendixson ranks.

Recall that in our case any minimal tiling is periodic
(no strictly quasiperiodic tiling appears in a countable setting
\cite{DBLP:journals/tcs/Durand99}).
The first step of the proof may thus be reformulated:
\begin{theorem}
\label{thm:peralorsperfini}
	If all tilings produced by a tile-set are periodic, then there are only
	finitely many of them.
\end{theorem}	

It is important to note that a compactness argument is not sufficient to prove this theorem,
there is no particular reason for a converging sequence of periodic tilings with
strictly increasing period to converge towards a non periodic tiling: 
there indeed exist such sequences with a periodic limit. 

\begin{proof}

\emph{We are in debt to an anonymous referee who simplified our original proof.}
	
	Suppose that a tile-set produces infinitely many tilings, but  only
	periodic ones.

	As the set of tilings is infinite and compact, one of them is obtained as a limit of
	the others: There exists a tiling $X$ and a sequence $X_i$ of distinct
	tilings such that $X_i \rightarrow X$.

	Now by assumption $X$ is periodic of period $p$ for some $p$. We may
	suppose that no $X_i$ has $p$ as a period.
	Denote by $M$ the pattern which is repeated periodically.

	$X_i \rightarrow X$ means that $X_i$ contains in its center a square
	of size $q(i) \times q(i)$ of copies of $M$, where $q$ is a growing function.
	
	For each $i$, consider the largest square of $X_i$ consisting only of
	copies of $M$. Such a largest square exists, as it is bounded by a
        period of $X_i$. Let $k$ be the size of this square.
		Now, the boundary of this square contains a $p \times p$ pattern which
	is not $M$ (otherwise this is not the largest square).
	
	By shifting $X_i$ so that this pattern is at the center, 
	we obtain a tiling $Y_i$ which contains a $p \times p$ pattern at the
        origin which is	not $M$ adjacent to a $k/2 \times k/2$
	square consisting of copies of $M$ in one of the four quarter planes.
	
        By taking a suitable limit of these $Y_i$, we will obtain a tiling which
	contains a $p \times p$ pattern which is not $M$ in its center
	adjacent to a quarter plane of copies of $M$.

	Such a tiling cannot be periodic.
\end{proof}

This proof does not assume that the set of forbidden patterns $\ForbPat$ is finite, therefore it is still
valid for any shift-invariant closed subset (subshift) of $Q^{\Z^2}$.

Now we prove stronger results about the Cantor-Bendixson rank of $\Tilings$.
Let $\alpha$ be the Cantor-Bendixson rank of $\Tilings$. 
Since $(\Tilings)^{(\alpha)} = \emptyset$, $\alpha$ cannot be a limit ordinal:
Suppose that it is indeed a limit ordinal, therefore $\bigcap_{\beta < \alpha} (\Tilings)^{(\beta)} = \emptyset$
is an empty intersection of closed sets in $Q^{\Z^2}$ 
therefore by compactness there exists $\gamma < \alpha$ such that
$\bigcap_{\beta < \gamma} (\Tilings)^{(\beta)} = \emptyset$ and therefore $\Tilings$ can not have rank $\alpha$.
Hence $\alpha$ is a successor ordinal, $\alpha =\beta + 1$.

However, we can refine this result : 

\begin{lemma}
\label{cbranklambdaplusdeux}
The rank of $\Tilings$ cannot be the successor of a limit ordinal.
\end{lemma}

\begin{proof}
Suppose that $\beta = \cup_{i<\omega}\beta_i$.
Since $(\Tilings)^{(\beta+1)} = \emptyset$, $(\Tilings)^{(\beta)}$ is finite
(otherwise it would have a non-isolated point by compactness), it contains only periodic tilings.

Let $p$ be the least common multiple of the periods of the tilings in
$(\Tilings)^{(\beta)}$. 
Let $M$ be the set of patterns of size $2p\times 2p$ that do not admit $p$ as a period.
Let $x_i$ be an element that is isolated in $(\Tilings)^{(\beta_i)}$.

As there is only a finite number of $p$-periodic tilings, we may suppose
w.l.o.g. that no $x_i$ admit $p$ as a period.

For any $i$, there exists a pattern of $M$ that appears in $x_i$.
Let $x'_i$ be the tiling with this pattern at its center.
By compactness, one can extract a limit $x'$ of the sequence $(x'_i)_{i\in\N}$, 
$x$' is by construction in $\cap_i (\Tilings)^{(\beta_i)} = \Tilings^{(\beta)}$.
However, $x'$ does not contain a $p-$periodic pattern at its center, that is
a contradiction.
\end{proof}
We write $\alpha = \lambda + 2$ the rank of $\Tilings$.

We already proved that there exists a non minimal tiling but this is not
sufficient to conclude that there exists a tiling at level $1$\footnote{We
actually can prove that the level $1$ exists: There is no infinite decreasing
chain whose lower bound is a periodic configuration}.
However, we achieve this as a corollary of the previous lemma:
$(\Tilings)^{(\lambda)}$ is infinite (otherwise $(\Tilings)^{(\lambda+1)}$ would be empty)
and contains a non periodic tiling by theorem \ref{thm:peralorsperfini}.
This non periodic tiling $c$ is not minimal (otherwise it would be strictly
quasiperiodic and then $\Tilings$ would not be countable).
Now $c$ is at level $1$ : any tiling less than $c$ is in
$(\Tilings)^{(\lambda + 1)}$ therefore periodic (hence minimal).

If a tiling $x$ is of \typea and is ranked, then it has a vector of
periodicity:
consider the pattern $P$ that isolates it in the last topological derivative of
$\Tilings$ that it belongs to.
Since $x$ is of \typea , this pattern appears twice in it, therefore there
exists a shift $\sigma$ such that 
$\sigma (x)$ contains $P$ at its center. $x = \sigma (x)$ because $P$ isolates $x$.


As any tiling of \typea has a vector of periodicity, it remains to prove that
$c$ is of \typea:
\begin{lemma}
$c$ is of \typea.
\end{lemma}

\begin{proof}
Suppose the converse : there exists a pattern $P$ that appears only once in $c$.
Considering the union of this pattern $P$ and a pattern that isolates
$c$, we may assume that $P$ isolates $c$.
$c$ has only a finite number of tilings smaller than itself:  they lie in
$\Tilings^{(\lambda +1)}$ which is finite, and are all periodic, say of
period $p$. As $P$ isolates $c$, none of these tilings contain $P$.

Consider the patterns of size $2p \times 2p$ of $T$ that are not $p-$periodic.
If those patterns can appear arbitrary far from $P$ then one can extract a
tiling from $c$ (thus smaller than $c$) that is not $p-$periodic and does not
contain $P$; this is not possible.

Therefore there is a pattern in $c$ that contains $P$ (thus appears only once) and
any other part of $c$ is $p-$periodic (one can gather all non $p-$periodic parts
of $c$ around $P$), as depicted in Fig.~\ref{c_type_b_alors_non_dnb:base}.

\begin{figure}[ht]
\begin{tabular}{c||c}
\subfigure[What we get : $c$ is periodic everywhere but at
$P$]{\label{c_type_b_alors_non_dnb:base}\includegraphics[width=.4\linewidth]{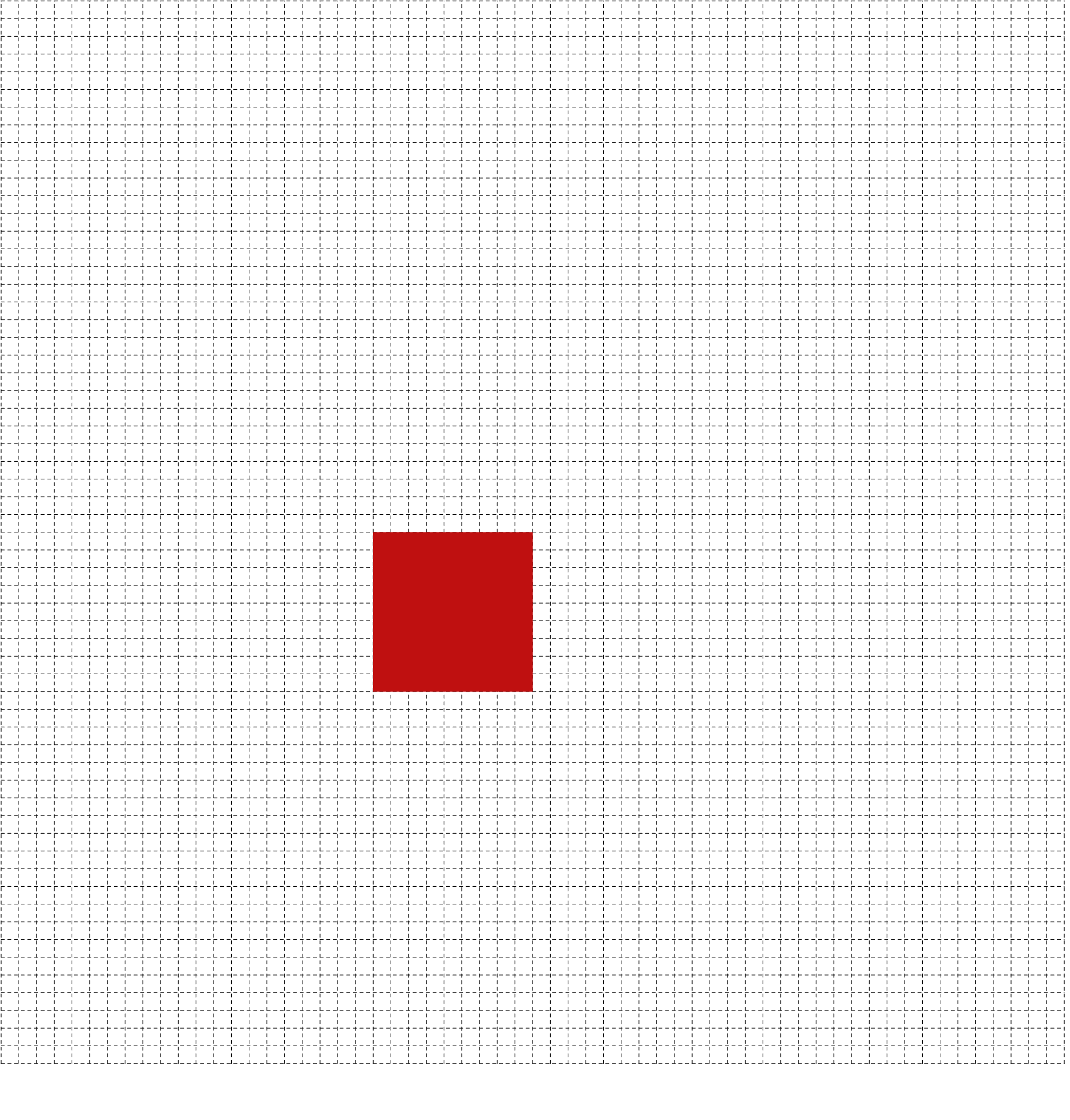}}
&
\subfigure[$P$ can appear at many different places since $c$ has periodic
patterns]{\label{c_type_b_alors_non_dnb:apres}\includegraphics[width=.4\linewidth]{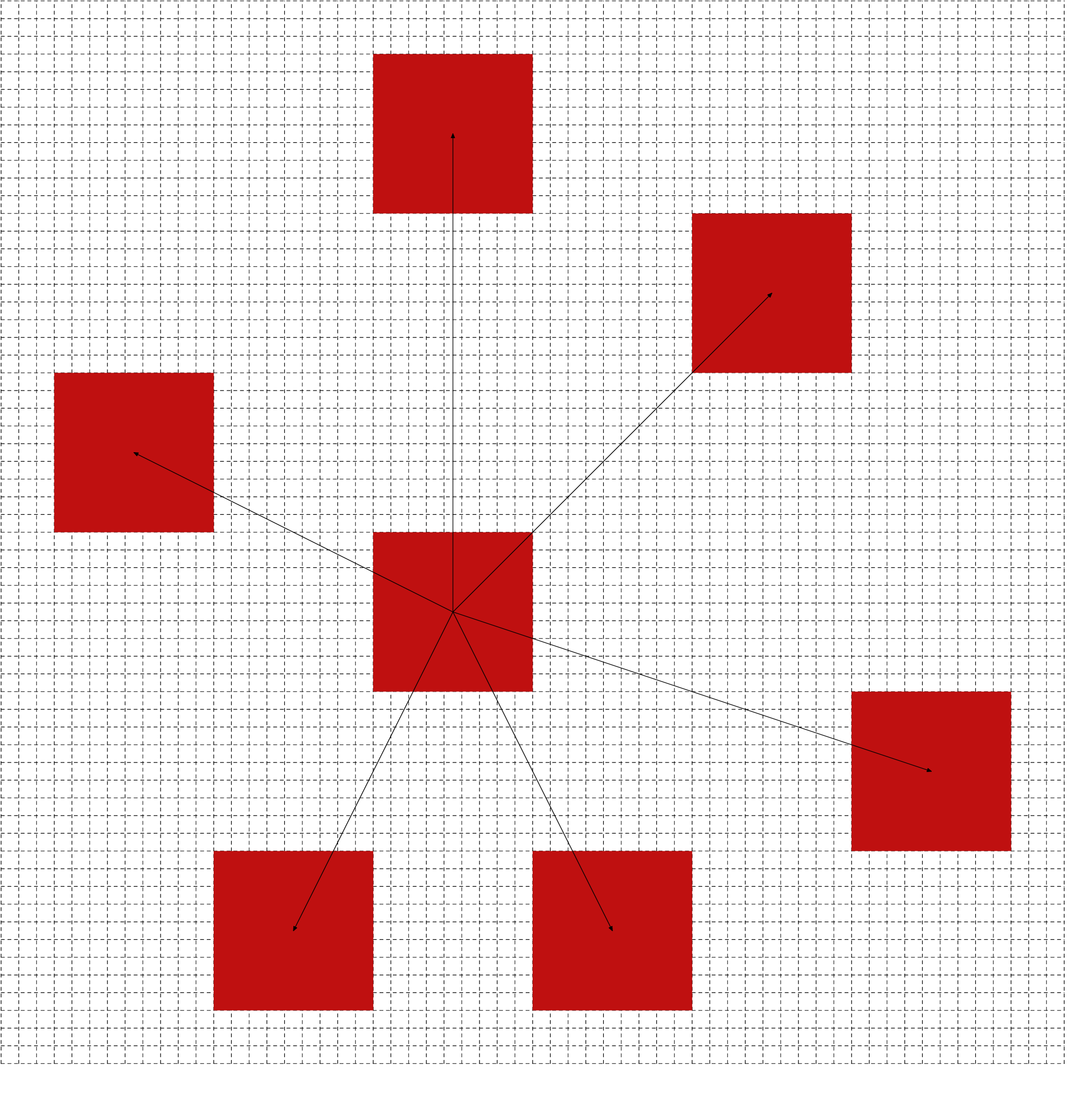}}
\end{tabular}
\caption{What can happen if $c$ is of \typeb ?}
\label{fig:c_type_b_alors_non_dnb}
\end{figure}

This non periodic part could also be inserted at infinitely many different
positions in $c$ since the tiling
rules are of bounded radius, as depicted in
Fig.~\ref{c_type_b_alors_non_dnb:apres}. Hence the number of tilings is not countable.
\end{proof}

$c$ is of \typea, $c$ is not periodic, $c$ has a vector of periodicity,
therefore our theorem \ref{thm:dnbalors1perstrict} holds :

\begin{theorem}
\label{thm:dnbalors1perstrict}
If $\tau$ is a tile-set that produces a countable number of tilings then it
produces a tiling with exactly one vector of periodicity.
\end{theorem}

\section{Open problems}

We are interested in proving more precise results for the order $\prec$
for a countable set of tilings: we wonder whether the order $\prec$ has at
most finitely many levels, as it is the case in Fig.~\ref{fig:treillis}. 
We know how to construct a tile-set so that the maximal level is any arbitrary
integer see \eg Fig.\ref{fig:levelgrandhasse} for level $3$.

We also intend to prove a similar result for uncountable sets of
tilings; the problem is that we are tempted to think that if the set
of tilings is uncountable, then a quasiperiodic tiling must appear.
However, this is not true: imagine a tile-set that admits a vertical 
line of white or black cells with red on the left and green on the 
right. The uncountable part is due to the vertical line that itself 
contains a quasiperiodic of dimension 1 but not of dimension 2. This tile-set
produces tilings that looks like $H$ in Fig.~\ref{fig:treillis}, except that the
vertical line can have two different colors without any constraint.

A generalization of lemma \ref{cbranklambdaplusdeux} would be to prove that the
Cantor-Bendixson rank of a countable set of tilings cannot be infinite; we know
how to construct sets of tilings that have an arbitrary large but finite Cantor-Bendixson
rank, but we do not know how to obtain a set of tilings of rank greater than $\omega$.

\bibliographystyle{plain}

\appendix


\begin{figure}[ht]
\centering

\includegraphics{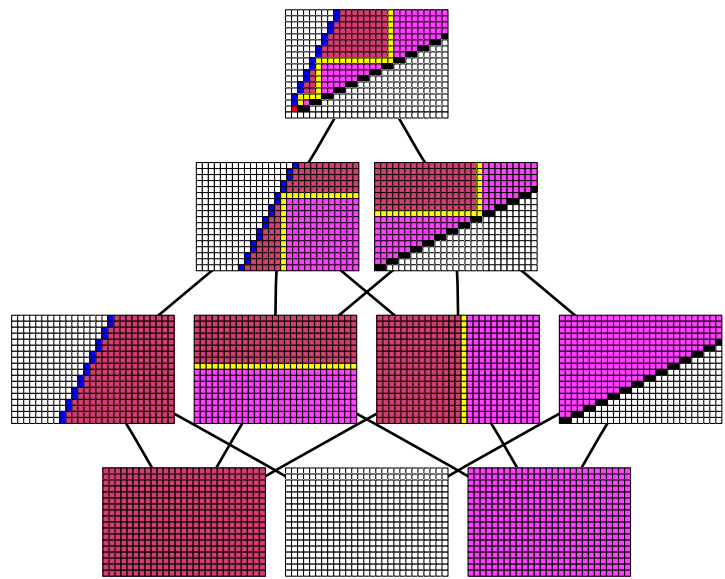}

\caption{An example of a tile-set that produces countably many tilings and a tiling
at level~3}
\label{fig:levelgrandhasse}
\end{figure}

\newpage
\null
\end{document}